\documentclass[conference, a4paper]{IEEEtran}
\usepackage{amsthm, amssymb, amsmath}
\usepackage{verbatim, float}
\usepackage{mathrsfs}
\usepackage{graphics}
\usepackage[usenames,dvipsnames]{pstricks}
\usepackage{epsfig}
\usepackage{pst-grad} 
\usepackage{pst-plot} 
\usepackage{cases}
\usepackage{algorithmic}
\usepackage{bm}

\usepackage[top=1.4cm, bottom=1.33cm, left=0.54in, right=0.54in]{geometry}
\usepackage{url}

\usepackage{hhline}
\usepackage{enumitem}

\usepackage{etoolbox}
\makeatletter
\patchcmd{\@makecaption}
  {\scshape}
  {}
  {}
  {}
\makeatother

\usepackage[singlelinecheck=false]{caption}
\usepackage{diagbox}

\theoremstyle{plain}
\newtheorem{theorem}{Theorem}
\newtheorem{lemma}{Lemma}

\newtheorem{corollary}{Corollary}

\theoremstyle{definition}
\newtheorem{definition}{Definition}

\DeclareMathOperator{\im}{im}

\newcommand{\B}{{\mathcal B}}
\newcommand{\C}{{\mathcal C}}

\newcommand{\F}{{\mathcal F}}

\newcommand{\K}{{\mathcal K}}

\newcommand{\R}{{\mathcal R}}

\renewcommand{\S}{{\mathcal S}}
\newcommand{\Sjjs}{{\mathcal S}_{j\to j^*}}

\newcommand{\Sjsjs}{{\mathcal S}_{j^*\to j^*}}

\newcommand{\Sjsj}{{\mathcal S}_{j^*\to j}}
\newcommand{\Sjo}{{\mathcal S}_{j\to 1}}
\newcommand{\Sto}{{\mathcal S}_{2\to 1}}
\newcommand{\Stho}{{\mathcal S}_{3\to 1}}

\newcommand{\WBjjs}{\mathcal{W}^{\B}_{j\to j^*}}

\newcommand{\WBjsj}{\mathcal{W}^{\B}_{j^*\to j}}

\newcommand{\bW}{\boldsymbol{W}}

\newcommand{\ba}{{\boldsymbol a}}
\newcommand{\bb}{{\boldsymbol b}}
\newcommand{\bu}{{\boldsymbol u}}

\newcommand{\bc}{{\boldsymbol c}}
\newcommand{\bcj}{{\boldsymbol c}_j}
\newcommand{\bcjs}{{\boldsymbol c}_{j^*}}
\newcommand{\bco}{{\boldsymbol c}_1}
\newcommand{\bcn}{{\boldsymbol c}_n}

\newcommand{\bw}{{\boldsymbol{w}}}

\newcommand{\bg}{{\boldsymbol{g}}}
\newcommand{\bgo}{{\boldsymbol{g}^{(1)}}}

\newcommand{\bgi}{{\boldsymbol{g}^{(i)}}}
\newcommand{\bgl}{{\boldsymbol{g}^{(\ell)}}}

\newcommand{\bhi}{{\boldsymbol{h}^{(i)}}}

\newcommand{\goj}{{\boldsymbol{g}^{(1)}_j}}

\newcommand{\gio}{{\boldsymbol{g}^{(i)}_1}}
\newcommand{\git}{{\boldsymbol{g}^{(i)}_2}}
\newcommand{\gith}{{\boldsymbol{g}^{(i)}_3}}
\newcommand{\gijs}{{\boldsymbol{g}^{(i)}_{j^*}}}
\newcommand{\gij}{{\boldsymbol{g}^{(i)}_j}}
\newcommand{\gin}{{\boldsymbol{g}^{(i)}_n}}

\newcommand{\glj}{{\boldsymbol{g}^{(\ell)}_j}}

\newcommand{\ff}{\mathbb{F}}
\newcommand{\ft}{\mathbb{F}_2}
\newcommand{\fq}{\mathbb{F}_q}
\newcommand{\ftl}{\mathbb{F}_{2^\ell}}

\newcommand{\fql}{\mathbb{F}_{q^{\ell}}}

\newcommand{\lra}{\leftrightarrow}

\newcommand{\Llra}{\Longleftrightarrow}

\newcommand{\al}{\alpha}

\newcommand{\bal}{\bm{\alpha}}
\newcommand{\bbe}{\bm{\beta}}

\newcommand{\bom}{\bm{\omega}}
\newcommand{\bkp}{\bm{\kappa}}
\newcommand{\bga}{\bm{\gamma}}

\newcommand{\bxi}{\bm{\xi}}

\newcommand{\wgb}{\bm{w}^{\bm{\gamma},\mathcal{B}}}

\newcommand{\wgtb}{\bm{w}^{\bm{\gamma}_t,\mathcal{B}}}

\newcommand{\wgib}{\bm{w}^{\bm{\gamma}_i,\mathcal{B}}}

\newcommand{\wxjgob}{\bm{w}^{\bxi^{j-1}\bm{\gamma}_1,\mathcal{B}}}
\newcommand{\wxjgtwb}{\bm{w}^{\bxi^{j-1}\bm{\gamma}_2,\mathcal{B}}}
\newcommand{\wxjgsb}{\bm{w}^{\bxi^{j-1}\bm{\gamma}_s,\mathcal{B}}}
\newcommand{\wxjgtb}{\bm{w}^{\bxi^{j-1}\bm{\gamma}_t,\mathcal{B}}}

\newcommand{\supp}{{\sf supp}}

\newcommand{\weight}{{\mathsf{wt}}}
\newcommand{\tr}{\mathsf{Tr}}

\newcommand{\spn}{\mathsf{span}_F}

\newcommand{\cL}{{\mathscr{L}}}
\newcommand{\cI}{{\mathscr{I}}}
\newcommand{\define}{\stackrel{\mbox{\tiny $\triangle$}}{=}}

\newcommand{\Cd}{\mathcal{C}^\perp}
\newcommand{\rsk}{\text{RS}(A,k)}

\newcommand{\rsnk}{\text{RS}(A,n-k)}

\begin{document}

\title{On the I/O Costs of Some Repair Schemes for Full-Length Reed-Solomon Codes
\thanks{
H. Dau is with the Department of Electrical and Computer System Engineering, Faculty of Engineering, Monash University, 14 Alliance Lane, Clayton, Victoria 3800, Australia. Email: hoang.dau@monash.edu.
I. Duursma is with the Departments of Mathematics, and also with the Coordinated Science Laboratory, University of Illinois at Urbana-Champaign, 1409 W. Green St, Urbana, IL 61801, USA. Email: duursma@illinois.edu.} 
}
\author{
Hoang Dau\\ Dept. ECSE, Monash University\\ Email: hoang.dau@monash.edu
\and Iwan Duursma\\ Dept. Mathematics, UIUC\\ Email: duursma@illinois.edu
\and Hien Chu\\ Dept. Mathematics, HCMUE\\Email: hienchu.1610@gmail.com
}
\date{}
\maketitle
\pagestyle{empty}

\begin{abstract}
\emph{Network transfer} and \emph{disk read} are the most time consuming operations in the repair process for node failures in erasure-code-based distributed storage systems. 
Recent developments on Reed-Solomon codes, the most widely used erasure codes in practical storage systems, have shown that efficient repair schemes specifically tailored to these codes can significantly reduce the \emph{network bandwidth} spent to recover single failures.
However, the \emph{I/O cost}, that is, the number of disk reads performed in these repair schemes remains largely unknown.
We take the first step to address this gap in the literature by investigating the I/O costs of some existing repair schemes for full-length Reed-Solomon codes. 
\end{abstract}

\section{Introduction}
\label{sec:intro}

Reed-Solomon (RS) codes~\cite{ReedSolomon1960}, although widely used as erasure codes to protect distributed storage systems (DSS) from frequent node failures, were believed to have very poor performance in repairing \emph{single} failures with respect to the \emph{repair bandwidth}. 
In the conventional/naive repair scheme for RS codes, the whole file has to be retrieved in order to repair just one lost data chunk. 
This drawback of RS codes led to the proposals of several other repair-efficient families of erasure codes such as regenerating codes~\cite{Dimakis_etal2007, Dimakis_etal2010, Dimakis_etal2010_survey} and locally repairable codes~\cite{OggierDatta2011,GopalanHuangSimitciYekhanin2012,PapailiopoulosDimakis2012}.

Despite the introduction of all of those new codes, RS codes remain to be the most popular codes in practice thanks to numerous inherent advantages, including optimal storage overhead, widest range of code parameters, and simple implementation. They are core components of major distributed storage systems such as Google's Colossus, Quantcast File System, Facebook's f4, Yahoo Object Store, Baidu's Atlas, Backblaze's Vaults, and Hadoop Distributed File System (see~\cite[Tab. I]{DauDuursmaKiahMilenkovic2016}).

In a recent line of research on repairing RS codes~\cite{Shanmugam2014, GuruswamiWootters2016, YeBarg_ISIT2016, DauMilenkovic2017, DuursmaDau2017, ChowdhuryVardy2017, TamoYeBarg2017, DauDuursmaKiahMilenkovic2016, DauDuursmaKiahMilenkovicTwoErasures2017, BartanWoottersMultipleFailures2017, YeBargMultipleErasures2017, LiWangJafarkhani-Allerton-2017}, it has been shown that with carefully crafted repair schemes, the repair bandwidth can be significantly reduced  for several families of RS codes. 
In this work, instead of focusing on the repair bandwidth, we investigate another important performance criterion for RS codes during the recovery process, that is, the (read) \emph{I/O cost} of the repair schemes\footnote{As reported in~\cite{MitraPantaRaBagchi2016}, network transfer and disk read constitute more than 98\% of the total reconstruction time in the Quantcast File System. It was also observed in another study~\cite{KhanBurnsPlankPierceHuang2012} that disk read always takes at least nine times longer than computation during repair or degraded read.}.
The I/O cost of a repair scheme is defined as the total amount of information that needs to be read from the disks located at the helper nodes during the repair of one failed node. 
The open question of how well RS codes perform when taken into account the I/O cost was originally raised by Guruswami and Wootters~\cite{GuruswamiWootters2016}. 

\begin{figure}[tb]
\centering
\includegraphics[scale=1]{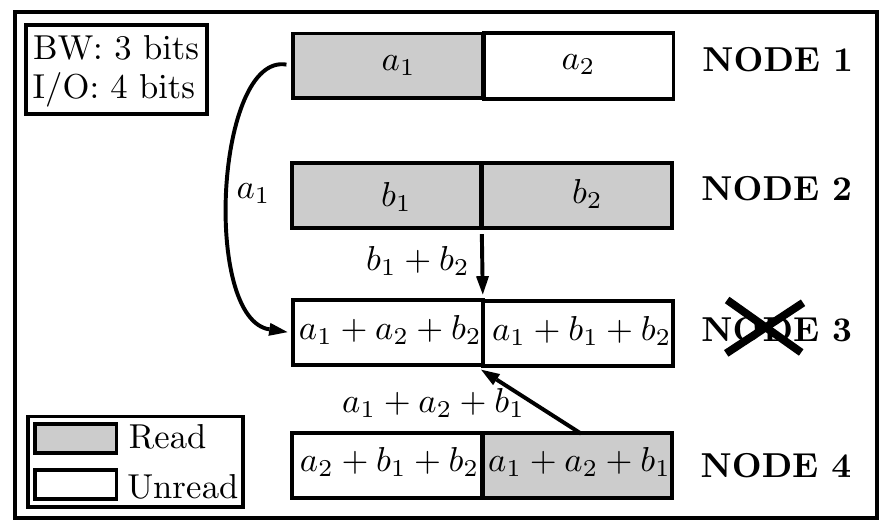}
\caption{Illustration of repair bandwidth and I/O cost during the repair process of one node failure in a DSS based on a $[4,2]_4$ Reed-Solomon code (see~\cite[Ex. 1]{DauDuursmaKiahMilenkovic2016} for its construction).
While the repair bandwidth is optimal (three bits), the I/O cost is four bits, which is as expensive as reading the whole file.
}
\label{fig:toy_example}
\end{figure}

To motivate the study of the I/O cost for RS codes, let us consider the toy example in Fig.~\ref{fig:toy_example}. 
The $4$-node storage system employs a $[4,2]$ RS codes over $\ff_4$ to store the file $(\ba,\bb) \in \ff_4^2$.
To reconstruct the two bits stored at Node~3 in a bandwidth-optimal way, the replacement node may contact three available nodes and downloads one bit of data from each. This results in a repair bandwidth of three bits, saving one bit compared to the conventional scheme, in which the replacement node contacts two nodes and downloads two bits from each.
However, the number of bits being read from the three nodes is four, which is the same as the file size. Thus, in terms of I/O cost, this bandwidth-efficient repair scheme is \emph{as expensive as the conventional repair scheme}.  
This observation raises an immediate question: when does this scenario happen?

\textbf{Our contribution.} We show that the bandwidth-optimal repair schemes proposed in~\cite{GuruswamiWootters2016, DauMilenkovic2017}, when applied to certain families of full-length RS codes, incur an I/O cost as high as that of the naive repair scheme (Section~\ref{sec:rotational}). 
We also prove that such a high I/O cost is a \emph{necessary} price to pay for the optimal bandwidth when the base field is $\ft$ and the code has two parities (Section~\ref{sec:converse}).

\section{Preliminaries}
\label{sec:pre}

Let $[n]$ denote the set $\{1,2,\ldots,n\}$. Let $F = \fq$ be the finite field of $q$ elements, for some prime power $q$. Let $E = \fql$ be an extension field of $F$, where $\ell \geq 1$, and let $E^* = E \setminus \{0\}$. 
We refer to the elements of $E$ as \emph{symbols} and the elements of $F$ as \emph{sub-symbols}. The field $E$ may also be viewed as a vector space of dimension $\ell$ over $F$, i.e. $E \cong F^\ell$, and hence each symbol in $E$ may be represented as a vector of length $\ell$ over $F$. 
We use $\spn(U)$ to denote the $F$-subspace of $E$ spanned by a set of elements $U$ of $E$. 
The (field) trace of any symbol $\bal \in E$ over $F$ is defined to be 
$\mathsf{Tr}_{E/F}(\bal) = \sum_{i = 0}^{\ell-1} \bal^{q^i}$. When clear from the context, we omit the subscript $E/F$. 
The support of a vector $\bu = (u_1,\ldots,u_\ell)$, denoted $\supp(\bu)$, is the set $\{j \colon u_j \neq 0\}$.
The (Hamming) weight of $\bu$, denoted $\weight(\bu)$, is $|\supp(\bu)|$.
The support of a set of vectors $U$ is $\supp(U) \define \cup_{\bu \in U} \supp(\bu)$.
A \emph{linear $[n,k]$ code} $\C$ over $E$ is an $E$-subspace of $E^n$ of dimension $k$. Each element of a code is referred to as a \emph{codeword}. 
The \emph{dual} $\Cd$ of a code $\C$ is the orthogonal complement of $\C$ in $E^n$ and has dimension $r = n - k$. 

\begin{definition} 
\label{def:RS}
Let $E[x]$ denote the ring of polynomials over $E$. A Reed-Solomon code $\rsk \subseteq E^n$ of dimension $k$ over a finite
field $E$ with evaluation points $A=\{\bal_j\}_{j=1}^n \subseteq E$
is defined as \vspace{-3pt} 
\[
\rsk = \Big\{\big(f(\bal_1),\ldots,f(\bal_n)\big) \colon f \in E[x],\ \deg(f) < k \Big\}. 
\]
\end{definition} 
The Reed-Solomon code is \emph{full length} if $n = |E|$. It is well known that the dual of a full-length Reed-Solomon code $\rsk$ is another Reed-Solomon code $\rsnk$~(as a corollary of~\cite[Chp.~10, Thm.~4]{MW_S}). 

\textbf{Trace repair framework.}
First, note that each symbol in $E$ can be recovered from its $\ell$ independent traces. More precisely, given a basis $\{\bbe_i\}_{i=1}^\ell$ of $E$ over $F$, any $\bal \in E$ can be uniquely determined given the values of $\tr(\bbe_i\,\bal)$ for $i\in [\ell]$, i.e. $\bal = \sum_{i=1}^\ell\tr(\bbe_i \bal)\bbe^\ast_i$, where $\{\bbe^\ast_i\}_{i=1}^\ell$ is the dual (trace-orthogonal) basis of $\{\bbe_i\}_{i=1}^\ell$ (see, e.g.~\cite[Ch.~2, Def.~2.30]{LidlNiederreiter1986}).

Let $\C$ be an $[n,k]$ linear code over $E$ and $\Cd$ its dual. 
If $\bc = (\bco,\ldots,\bcn) \in \C$ and $\bg = (\bg_1,\ldots,\bg_n) \in \Cd$ then $\bc \cdot \bg = \sum_{j=1}^n \bcj\bg_j = 0$. Suppose $\bcjs$ is erased and needs to be recovered.
In the trace repair framework, choose a set of $\ell$ dual codewords $\bgo,\ldots,\bgl$ such that $\dim_F\big(\{\gijs\}_{i=1}^\ell\big) = \ell$. Since the trace is a linear map, we obtain the following $\ell$ equations \vspace{-5pt} 
\begin{equation}
\label{eq:repair_equations} 
\tr\big(\gijs \bcjs\big) = -\sum_{j \neq j^*} \tr\big(\gij \bcj\big),\quad i \in [\ell]. \vspace{-5pt}
\end{equation} 
In order to recover $\bcjs$, one needs to retrieve sufficient information from $\{\bcj\}_{j \neq j^*}$ to compute the right-hand sides of \eqref{eq:repair_equations}. 
We define, for every $j \in [n]$, \vspace{-5pt}
\begin{equation}
\label{eq:S}
\Sjjs \define \spn\bigg(\left\{\goj,\ldots,\glj\right\}\bigg) \vspace{-5pt}
\end{equation}
and refer to $\Sjjs$ as a \emph{column-space} of the repair scheme when $j \neq j^*$. 
Then for each $j \neq j^*$, in order to determine $\tr(\gij \bcj)$ for all $i \in [\ell]$, 
it suffices to retrieve $\dim_F(\Sjjs)$ sub-symbols (in $F$) only. 
Indeed, suppose $\{\bg^{i_t}_j\}_{t=1}^s$ is an $F$-basis of $\Sjjs$, then by 
retrieving just $s$ traces $\tr(\bg^{i_1}_j \bcj),\ldots,\tr(\bg^{i_s}_j \bcj)$ of $\bcj$, 
all other traces $\tr(\gij \bcj)$ can be computed as $F$-linear combinations of those $s$ traces without any knowledge of $\bc$. Finally, since $\{\gijs\}_{i=1}^\ell$ is $F$-linearly independent, $\bcjs$ can be recovered from its $\ell$ corresponding traces on the left-hand side of \eqref{eq:repair_equations}. 
We refer to such a scheme as a \emph{repair scheme based on} $\{\bgi\}_{i=1}^\ell$.
It was known that this type of repair schemes includes every possible linear repair scheme for RS codes~\cite{GuruswamiWootters2016}.  

\begin{lemma}[Guruswami-Wootters~\cite{GuruswamiWootters2016}] 
\label{lem:GW}
Suppose $E = \fql$, $F = \fq$, $\C$ is an $[n,k]$ linear code over $E$ and $\Cd$ is its dual. 
The repair scheme for $\bcjs$ based on $\ell$ dual codewords $\bgo,\ldots,\bgl$, where $\dim_F\big(\{\gijs\}_{i=1}^\ell\big) = \ell$, incurs a repair bandwidth of $\sum_{j \neq j^*} \dim_F(\Sjjs)$ sub-symbols in $F$, where $\Sjjs$ is defined as in \eqref{eq:S}.  
\end{lemma} 

\textbf{I/O Cost of a Repair Scheme.}
Let $\B = \{\bbe_i\}_{i=1}^\ell$ be an $F$-basis of $E$. For each $\bal \in E$, we may write $\bal = \sum_j \al_j \bbe_j$, where $\al_j \in F$. The vector $(\al_1,\ldots,\al_\ell)\in F^\ell$ is the vector representation of $\bal$ with respect to the basis $\B$.
We often write $\bal = (\al_1,\ldots,\al_\ell)_{\B}$ or just $\bal = (\al_1,\ldots,\al_\ell)$ for brevity.
We first define the I/O cost of a function and then proceed to describe the I/O cost of a repair scheme. The underlying assumption is that each sub-symbol $\al_i$ of $\bal$ can be read from the storage disk separately without accessing other sub-symbols. 

\begin{definition}[I/O cost of functions] 
\label{def:I/O}
The (read) I/O cost of a function $f(\cdot)$ with respect to a basis $\B$ is the minimum number of sub-symbols of $\bal \in E$ needed to compute $f(\bal)$.
The I/O cost of a set of functions $\F$ is the minimum number of sub-symbols of $\bal$ needed for the computation of $\{f(\bal) \colon f \in \F\}$. 
\end{definition} 

\begin{lemma} 
\label{lem:I/O_linear}
The following statements hold. 
\begin{itemize}[leftmargin=0.6cm]
	\item[(a)] The I/O cost of a linear function $f_{\bw}(\bal) \define \bw \cdot \bal = \sum_j w_j\al_j$ with respect to a basis $\B$ is $\weight(\bw) = |\supp(\bw)|$, where $\bw = (w_1,\ldots,w_\ell) \in E$.
	\item[(b)] The I/O cost of a set of linear functions $\bw_1 \cdot \bal, \ldots, \bw_s \cdot \bal$ with respect to $\B$ is $|\cup_{j = 1}^{s} \supp(\bw_j)|$.
	\item[(c)] The I/O cost of the trace functional $\tr_{\bga}(\cdot)$, defined by $\tr_{\bga}(\bal) \define \tr(\bga\bal)$, with respect to $\B$ is $\weight\big( \wgb \big)$, where 
	\begin{equation} 
\label{eq:wbg}
\wgb \define \big(\tr(\bga\bbe_1),\ldots,\tr(\bga\bbe_\ell)\big). 
\end{equation} 
	\item[(d)] The I/O cost of the set of trace functionals $\{\tr_{\bga}(\cdot) \colon \bga \in \Gamma\}$ with respect to $\B$ is $|\cup_{\bga \in \Gamma} \supp(\wgb)|$. 
\end{itemize}
\end{lemma} 
\begin{proof} 
The statements (a) and (b) follow directly from Definition~\ref{def:I/O}. 
Statement (c) holds as $\tr_{\bga}(\bal) = \sum_i \tr(\bga\bbe_i)\al_i = \wgb \cdot \bal$. The last statement follows from (b) and (c).  
\end{proof} 

The following lemma is due to the linearity of trace. 

\begin{lemma} 
\label{lem:wgb_basis}
Let $\Gamma \define \spn\big(\{\bga_i\}_{i=1}^\ell\big)$, $\{\bga'_i\}_{i=1}^s$ be an $F$-basis of $\Gamma$, $\wgb$ defined as in \eqref{eq:wbg}. The following statements hold.
\begin{itemize}[leftmargin=0.6cm]
	\item[(a)] If $\bga\hspace{-0.1cm} =\hspace{-0.1cm} \sum_{i=1}^\ell a_i\bga_i$, for $a_i \in F$, then $\wgb\hspace{-0.1cm} =\hspace{-0.1cm} \sum_{i=1}^\ell a_i\wgib$, and therefore, $\supp(\wgb) \subseteq \cup_{i=1}^\ell \supp(\wgib)$. 
	\item[(b)] $\cup_{i = 1}^\ell \supp(\wgib)\hspace{-0.1cm} =\hspace{-0.1cm} \cup_{\bga \in \Gamma} \supp(\wgb)\hspace{-0.1cm} =\hspace{-0.1cm} \cup_{t = 1}^s \supp(\bw^{\bga'_i,\B}).$
\end{itemize}
\end{lemma} 

The I/O cost of the repair scheme based on a set of dual codewords $\{\bgi\}_{i=1}^\ell$ is the minimum number of sub-symbols of $\bcj$'s needed in the computation of the right-hand sides of \eqref{eq:repair_equations}. We provide the formal definition below. 

\begin{definition}[I/O cost of a repair scheme] 
The I/O cost of the repair scheme based on a set of dual codewords $\{\bgi\}_{i=1}^\ell$ with respect to a basis $\B$ is the sum of the I/O costs of the sets of trace functionals $\F_j = \big\{\tr_{\gij}(\cdot)\big\}_{i=1}^\ell$, $j \in [n] \setminus \{j^*\}$.  
\end{definition} 

\begin{lemma}
\label{lem:I/O} 
Suppose $\bcjs$ is lost and needs to be recovered. The I/O cost of the repair scheme based on $\ell$ dual codewords $\bgi = \big(\gio,\ldots,\gin\big)$, $i \in [\ell]$, with respect to a basis $\B$ is
\[
\sum_{j \neq j^*} \big|\supp\big(\WBjjs\big)\big|, 
\]
where $\wgb$ is defined as in \eqref{eq:wbg} and
\begin{equation}
\label{eq:WB}
\WBjjs \define \left\{\wgb \colon \bga \in \S_{j\to j^*}\right\}.
\end{equation}
\end{lemma} 
\begin{proof}
By Lemma~\ref{lem:I/O_linear}~(d), the I/O cost of the set of trace functionals $\F_j = \big\{\tr_{\gij}(\cdot)\big\}_{i=1}^\ell$ is $\big|\cup_{i=1}^\ell \supp\big(\bw^{\gij,\B}\big)\big|$ for each $j \neq j^*$. 
According to Lemma~\ref{lem:wgb_basis}~(b) and the definition of $\Sjjs$, this quantity is equal to $|\cup_{\bga \in \Sjjs} \supp(\wgb)|$, which is the same as $|\supp(\WBjjs)|$.
The lemma follows by summing up the I/O costs of the sets $\F_j$, $j \in [n] \setminus \{j^*\}$.  
\end{proof} 


\section{I/O Cost of Rotational Repair Schemes}
\label{sec:rotational}

\subsection{Fixed Basis}

Suppose $\bcjs$ is lost and needs to be recovered, and a basis $\B$ is fixed for all storage nodes, which is the usual situation in practice. 
The ultimate goal is to find Pareto solutions to the following \emph{multiobjective} optimization problem.\vspace{-4pt}
\begin{equation}
\label{eq:mop1}
\hspace{-0.05cm}\min_{\genfrac{}{}{0pt}{}{\bgo,\ldots,\bgl \in \Cd}{\dim_F(\Sjsjs) = \ell}}\hspace{-0.2cm} \bigg(\sum_{j \neq j^*}\hspace{-0.1cm} \dim_F(\Sjjs),\hspace{-0.05cm} \sum_{j \neq j^*}\hspace{-0.1cm}\big|\supp\big(\WBjjs\big)\big| \bigg),
 \vspace{-4pt}
\end{equation}
where $\Sjjs$ and $\WBjjs$ are defined in \eqref{eq:S} and \eqref{eq:WB}, respectively. 

Note that if we ignore the second objective function on the I/O cost, the problem reduces to the previously studied problem of minimizing the repair bandwidth. Ignoring the first objective function instead, it reduces to the one minimizing the I/O cost only. 
The problem \eqref{eq:mop1}, which seeks to minimize a multiobjective function consisting of bandwidth and I/O cost, appears to be challenging even for very particular sets of code parameters. Hence, we start with a simpler task: to study the I/O costs of those repair schemes that achieve optimal repair bandwidth for certain families of full-length RS codes. 

\begin{definition}[Rotational repair scheme] 
Let $n = |E| = q^\ell$. The repair scheme for $\bcjs$ based on a set of dual codewords $\{\bgi\}_{i=1}^\ell$ is called \emph{rotational} if there exists an $F$-subspace $\S$ of $E$ 
such that $\Sjjs = \rho_j\S$ for every $j \neq j^*$ and moreover, $\{\rho_j\}_{j \neq j^*} = E^*$.   
\end{definition} 

In other words, a repair scheme for $\bcjs$ is rotational if each of its column-space $\Sjjs$, $j \neq j^*$, is a translate of a common $F$-subspace $\S$ of $E$ with a different multiplier.
As a consequence, in a rotational repair scheme, every column-space has the same $F$-dimension. This common dimension is referred to as the \emph{column-dimension} of the rotational repair scheme.
Before presenting our main theorem on the I/O cost of a rotational repair scheme for full-length RS codes, a few auxiliary lemmas are needed. 

\begin{lemma} 
\label{lem:intersection}
Let $K \define \ker(\tr(\cdot)) = \{\bal \in E \colon \tr(\bal) = 0\}$.
If $\{\bga_t\}_{t=1}^s \subseteq E$ is $F$-linearly independent then\vspace{-5pt}
\[
\dim_F\big( \cap_{t=1}^s K / \bga_t \big) = \ell - s. \vspace{-3pt}
\]
\end{lemma}
\begin{proof} 
Set $K_s = \cap_{t=1}^s K/\bga_t$. 
We prove that $\dim_F(K_s) = \ell - s$ by induction in $s$.
Clearly, $\dim_F(K_1) = \dim_F(K/\bga_1) = \dim_F(\ker(\tr)) = \ell-1$. 
Suppose that $\dim_F(K_{s-1}) = \ell - s + 1$.
We aim to show that $\dim_F(K_s) = \ell - s$. 

Let $\cL$ be the vector space of all $q^\ell$ linear mappings from $E$ to $F$
and define an equivalence relation $\lra$ on $\cL$ as follows: $f\lra g$ if $f|_{K_{s-1}} \equiv g|_{K_{s-1}}$. As $\dim_F(K_{s-1}) = \ell-s+1$, there are precisely $q^{\ell-s+1}$ distinct linear mappings from $\K_{s-1}$ to $F$. Therefore, there are $q^{\ell-s+1}$ equivalence classes with respect to the relation $\lra$. Each of such classes contains $q^{s-1}$ mappings. 
Hence, there are precisely $q^{s-1}$ linear mappings $f \in \cL$ satisfying $f|_{K_{s-1}} \equiv 0$, which consitute the equivalence class $E_0$ containing the trivial mapping $f \equiv 0$. 

One can easily verify that $\tr_{\bga}(\cdot) \in E_0$ for every $\bga \in \spn\big(\{\bga_t\}_{t=1}^{s-1}\big)$. As $\{\bga_t\}_{t=1}^{s-1}$ is $F$-linearly independent, there are $q^{s-1}$ such trace functionals. Therefore, $E_0 = \left\{\tr_{\bga}(\cdot) \colon \bga \in \spn\big(\{\bga_t\}_{t=1}^{s-1}\big)\right\}$. Hence, for $\bga_s \notin \spn\big(\{\bga_t\}_{t=1}^{s-1}\big)$, we have $\tr_{\bga_s}(\cdot) \notin E_0$. That implies $\ker(\tr_{\bga_s}) \not\supseteq K_{s-1}$. Equivalently, $K/\bga_s \not\supseteq K_{s-1}$. Thus,\vspace{-5pt}
\[
\dim_F(K_s)\hspace{-0.06cm} =\hspace{-0.06cm} \dim_F\bigg( K_{s-1} \bigcap \dfrac{K}{\bga_s} \bigg)\hspace{-0.06cm} <\hspace{-0.06cm} \dim_F(K_{s-1})\hspace{-0.06cm} =\hspace{-0.06cm} \ell-s+1, \vspace{-5pt}
\]
which implies that $\dim_F(K_s) \leq \ell-s$. 
To conclude, it remains to show that $\dim_F(K_s) \geq \ell - s$. 
Indeed, consider the linear mapping $\sigma \colon K_{s-1} \to F$, defined as $\sigma(\bkp) = \tr(\bga_s \bkp)$ for $\bkp \in K_{s-1}$. 
Then 
$\ker(\sigma) = K_{s-1} \cap (K/\bga_s) = K_s$.  
Therefore,  \vspace{-5pt}
\[
\dim_F(K_s) = \dim_F(\ker(\sigma)) \geq \dim_F(K_{s-1}) - 1 = \ell - s.   \vspace{-5pt}
\] 
This completes the proof. 
\end{proof}  

\begin{lemma} 
\label{lem:main}
Suppose $\{\bga_t\}_{t=1}^s \subseteq E$ is an $F$-linearly independent set and $\bxi$ is a primitive element of $E$. Set \vspace{-5pt}
\[
b_j = 
\begin{cases}
0,& \text{ if } \tr(\bga_1\bxi^j) = \cdots = \tr(\bga_s\bxi^j) = 0,\\
1,& \text{ otherwise.}
\end{cases} \vspace{-5pt}
\]
Then we have $\sum_{j = 0}^{q^\ell-2} b_j = q^\ell - q^{\ell-s}$. 
\end{lemma}
\begin{proof} 
It suffices to show that $|\{j \colon b_j\hspace{-0.05cm} =\hspace{-0.05cm} 0\}|\hspace{-0.05cm} =\hspace{-0.05cm} q^{\ell-s}-1$.
We have \vspace{-7pt}
\[
\begin{split}
b_j = 0 &\Llra \tr(\bga_1\bxi^j) = \cdots = \tr(\bga_s\bxi^j) = 0 \Llra \bxi^j \in \bigcap_{t=1}^s \dfrac{K}{\bga_t}, 
\end{split}  \vspace{-20pt}
\]
where $K = \ker(\tr)$. 
According to Lemma~\ref{lem:intersection}, \vspace{-5pt}
\[
\dim_F \bigg(\bigcap_{t=1}^s \dfrac{K}{\bga_t}\bigg) = \ell-s. \vspace{-5pt}
\] 
Therefore, \vspace{-10pt}
\[
|\{j \colon b_j = 0\}| = \left|\left\{j \colon \bxi^j \in \bigcap_{t=1}^s \dfrac{K}{\bga_t}\right\}\right| = q^{\ell-s}-1,  \vspace{-5pt}
\]
as desired. The proof follows. 
\end{proof} 

\begin{theorem}
\label{thm:fixed_basis}
The I/O cost of a rotational repair scheme with column-dimension $s$ for a full-length Reed-Solomon code over $\fql$ is $\ell(q^\ell-q^{\ell-s})$.
\end{theorem}
\begin{proof}
Given a rotational repair scheme based on $\{\bgi\}_{i=1}^\ell$ with column-dimension $s$, according to Lemma~\ref{lem:I/O}, we need to show that $\sum_{j \neq j^*} |\supp(\WBjjs)| = \ell(q^\ell - q^{\ell-s})$. 

To simplify the notation, without loss of generality, we may assume that $j^* = n$ and $\Sjjs = \bxi^{j-1}\S$, $j \in [n-1]$, where $\S$ is an $s$-dimensional $F$-subspace of $E$ and $\bxi$ is a primitive element of $E$. 
Let $\{\bga_t\}_{t=1}^s$ be an $F$-basis of $\S$. Then $\{\bxi^{j-1}\bga_t\}_{t=1}^s$ forms an $F$-basis of $\Sjjs$ for every $j \in [n-1]$. 
Therefore, by Lemma~\ref{lem:wgb_basis}~(b), we have \vspace{-7pt}
\[
\supp(\WBjjs) =\hspace{-0.2cm} \bigcup_{\bga \in \Sjjs}\supp\big(\wgb\big) = \bigcup_{t=1}^s \supp\big(\wxjgtb\big). \vspace{-7pt}
\]
Recall that \vspace{-5pt}
\[
\wxjgtb = \big(\tr(\bxi^{j-1}\bga_t\bbe_1),\ldots,\tr(\bxi^{j-1}\bga_t\bbe_\ell)\big) \in F^\ell.  \vspace{-5pt}
\]
Then $|\supp(\WBjjs)|$ is precisely the number of nonzero columns in the $s\times \ell$ maxtrix $\bW_j$ whose rows are $\wxjgtb$, $t \in [s]$, \vspace{-5pt}
\[
\begin{split}
\bW_j &\define 
\begin{pmatrix}
\wxjgob \\
\hline
\wxjgtwb \\
\hline
\vdots \\
\hline
\wxjgsb
\end{pmatrix}\\
&=
\begingroup 
\setlength\arraycolsep{1pt}
\begin{pmatrix}
\tr(\xi^{j-1}\bga_1\bbe_1) & \cdots & \tr(\xi^{j-1}\bga_1\bbe_i) & \cdots & \tr(\xi^{j-1}\bga_1\bbe_\ell)\\ 
\tr(\xi^{j-1}\bga_2\bbe_1) & \cdots & \tr(\xi^{j-1}\bga_2\bbe_i) & \cdots & \tr(\xi^{j-1}\bga_2\bbe_\ell)\\ 
\vdots & \vdots & \vdots & \vdots & \vdots \\
\tr(\xi^{j-1}\bga_s\bbe_1) & \cdots & \tr(\xi^{j-1}\bga_s\bbe_i) & \cdots & \tr(\xi^{j-1}\bga_s\bbe_\ell) 
\end{pmatrix}.
\endgroup
\end{split}
\]
Therefore, the I/O cost of the repair scheme is equal to the total number of nonzero columns in the matrices $\bW_1,\ldots,\bW_{n-1}$. Thus, setting \vspace{-5pt} 
\[
b_{i,j} =
\begin{cases}
0,& \text{ if } \tr(\bxi^{j-1}\bga_1\bbe_i) = \cdots = \tr(\bxi^{j-1}\bga_s\bbe_i) = 0,\\
1,& \text{ otherwise,}
\end{cases} \vspace{-5pt}
\]
the I/O cost of the repair scheme can be computed as \vspace{-5pt}
\[
\begin{split}
\sum_{j=1}^{n-1} |\supp(\WBjjs)| &= \sum_{j=1}^{n-1}\sum_{i=1}^\ell b_{i,j}
= \sum_{i=1}^\ell\sum_{j=1}^{n-1} b_{i,j}\\
&= \sum_{i=1}^\ell (q^\ell - q^{\ell-s}) = \ell(q^\ell - q^{\ell-s}),
\end{split} \vspace{-5pt}
\]
where the third equality follows by applying Lemma~\ref{lem:main} to the $F$-linearly independent set $\{\bga_t\bbe_i\}_{t=1}^s$ and by setting $b_{j} \define b_{i,j+1}$, $j = 0,\ldots,q^\ell-2 = n-2$. This completes the proof.  
\end{proof}

The bandwidth-optimal repair schemes for full-length RS codes proposed by Dau and Milenkovic~\cite{DauMilenkovic2017}, one of which directly generalizes the scheme proposed by Guruswami and Wootters~\cite{GuruswamiWootters2016}, are both rotational. As a consequence, their I/O costs can be explicitly determined. We conclude that although these schemes achieve optimal repair bandwidth for RS codes, the I/O cost required is as high as that of the naive repair. 

\begin{corollary} 
The repair schemes for full-length Reed-Solomon codes with $n = q^\ell$ and $r = n - k = q^m$, $1 \leq m < \ell$, proposed in~\cite{DauMilenkovic2017}, have the I/O cost $k\ell$ sub-symbols in $F$.
\end{corollary} 
\begin{proof} 
There are two repair schemes presented in~\cite{DauMilenkovic2017}, both of which are based on the subspace polynomial $L_{W}(x) = \prod_{\bom \in W} (x-\bom)$, where $W$ is an $m$-dimensional $F$-subspace of $E$. 
In their Construction~III, the set of dual codewords used to repair $\bcjs$ is given below, where $\{\bbe_i\}_{i=1}^\ell$ is an $F$-basis of $E$.
\[
\bgi\hspace{-0.06cm} =\hspace{-0.06cm} \bigg(\hspace{-0.06cm}\frac{L_W(\bbe_i(\bal_1\hspace{-0.05cm} -\hspace{-0.05cm} \bal_{j^*}))}{\bal_1\hspace{-0.05cm} -\hspace{-0.05cm} \bal_{j^*}},\ldots,\frac{L_W\big(\bbe_i(\bal_n\hspace{-0.05cm} -\hspace{-0.05cm}\bal_{j^*})\big)}{\bal_n\hspace{-0.05cm} -\hspace{-0.05cm}\bal_{j^*}}\hspace{-0.06cm}\bigg),\hspace{-0.06cm}\ i \in [\ell].
\] 
Let $\S = \im(L_W)$, which is an $(\ell-m)$-dimensional $F$-subspace of $E$. The column-spaces in this repair scheme are $(j \neq j^*)$ 
\[
\Sjjs\hspace{-0.05cm} =\hspace{-0.05cm} \spn\bigg(\hspace{-0.07cm}\left\{\hspace{-0.05cm}\dfrac{L_W(\bbe_i(\bal_j\hspace{-0.05cm} -\hspace{-0.05cm} \bal_{j^*}))}{\bal_j\hspace{-0.05cm} -\hspace{-0.05cm} \bal_{j^*}}\colon i \in [\ell]\hspace{-0.05cm}\right\}\hspace{-0.07cm}\bigg)\hspace{-0.05cm}
=\hspace{-0.05cm} \dfrac{\S}{\bal_j\hspace{-0.05cm} -\hspace{-0.05cm}\bal_{j^*}}. 
\]
For the last equality, note that as $\{\bbe_i(\bal_j-\bal_{j^*})\}_{i=1}^\ell$ forms an $F$-basis of $E$ and $L_W$ is a linear mapping from $E$ to itself, the set $\big\{L_W\big(\bbe_i(\bal_j-\bal_{j^*})\big)\big\}_{i=1}^\ell$ indeed spans the subspace $\S = \im(L_W)$. 
As for $n = |E|$ we have \vspace{-3pt}
\[
\left\{1/(\bal_j-\bal_{j^*}) \colon j \in [n] \setminus \{j^*\} \right\} = E^*, \vspace{-3pt}
\]
the corresponding repair scheme is a rotational one with column-dimension $s = \dim_F(\S) = \ell-m$. Thus, according to Theorem~\ref{thm:fixed_basis}, the repair scheme in~\cite[Construction~III]{DauMilenkovic2017} has an I/O cost of 
\[
\ell\big(q^\ell - q^{\ell-(\ell-m)}\big) = \ell\big(q^\ell - q^m\big) = \ell(n-r) = k\ell.
\]
The same conclusion applies to the repair scheme in~\cite[Construction~II]{DauMilenkovic2017} using similar arguments. 
\end{proof} 

\subsection{Flexible Bases}

The choice of bases used to represent finite field elements, which clearly does not affect the repair bandwidth, may have an impact on the I/O cost of the repair scheme. For instance, suppose $j, j^* \in [n]$ such that $j \neq j^*$, and $\R_{j^*}$ is a repair scheme for $\bcjs$. Node~$j$ can easily choose a suitable basis $\B$ that minimizes the amount of data it needs to read according to $\R_{j^*}$ as follows. Let $\Sjjs$ be the column-space of $\R_{j^*}$ and $\{\bga_t\}_{t=1}^s$ one of its $F$-basis.
We can extend this basis of $\Sjjs$ to a basis of $E$, namely $\{\bga_i\}_{i=1}^\ell$, and select $\B = \{\bbe_i\}_{i=1}^\ell$ as its dual, i.e., $\tr(\bga_i\bbe_j) = 1$ if $i=j$ and $0$ otherwise. Then \vspace{-5pt} 
\[
\wgtb\hspace{-0.1cm} =\hspace{-0.1cm} \big(\tr(\bga_t\bbe_1),\ldots,\tr(\bga_t\bbe_\ell) \big)\hspace{-0.1cm} =\hspace{-0.1cm} \bm{e}_t\hspace{-0.1cm} =\hspace{-0.1cm} (\underbrace{0,\ldots,0,1}_{t},0,\ldots,0). \vspace{-5pt}
\]
The number of sub-symbols of $\bcj$ that Node~$j$ has to read is \vspace{-5pt}
\[
\big| \supp\big(\WBjjs\big) \big| = \big|\bigcup_{t=1}^s \supp\big(\wgtb\big) \big| = s = \dim(\Sjjs). \vspace{-5pt}
\]
Note that the I/O cost incurred at a particular node is always bounded from below by the bandwidth used at that node, i.e. $\big| \supp\big(\WBjjs\big) \big| \geq \dim(\Sjjs)$. Therefore, selecting this basis, Node~$j$ is able to minimize the I/O cost incurred in repairing Node~$j^*$. 
This particular choice of basis, however, may not work well for Node~$j$ in the repair process of other nodes. 
Therefore, given a collection of $n$ repair schemes for every node, one could seek to minimize the \emph{average} I/O cost at each storage node in the repair process of all other $n-1$ nodes. The average I/O cost of a collection of repair schemes is defined as follows. \vspace{-5pt} 
\begin{equation}
\label{eq:mop2}
\cI(\R) \define \dfrac{1}{n}\sum_{j=1}^n \min_{\B} \sum_{j^* \neq j} \left|\supp\big(\WBjjs \big)\right|, \vspace{-5pt}
\end{equation}
where $\WBjjs$ is defined as in \eqref{eq:WB} and the collection of repair schemes $\R = \{\R_{j^*}\}_{j^*=1}^n$ is given.
Given that bandwidth is usually the most expensive resource, we find it reasonable to start out with a collection of repair schemes that are bandwidth efficient and then proceed to optimize its average I/O cost. 

\begin{definition}[Symmetric repair schemes]
A collection of $n$ repair schemes $\R=\{\R_{j^*}\}_{j^*=1}^n$ is said to be \emph{symmetric} if $\Sjsj = \Sjjs$, for every $j \neq j^*$, $j,j^*\in[n]$. 
\end{definition} 

\begin{theorem} 
\label{thm:flexible_basis}
The average I/O cost of a symmetric collection of rotational repair schemes with column-dimension $s$ for a Reed-Solomon code of full length $q^\ell$ is $\ell(q^\ell-q^{\ell-s})$.
\end{theorem}
\begin{proof} 
As the collection is symmetric, $\Sjsj = \Sjjs$, which implies $\WBjsj = \WBjjs$. Hence, the total I/O cost incurred at Node~$j^*$ during the repair of all other nodes is
\[
\sum_{j \neq j^*} \big|\supp\big(\WBjsj\big)\big|\hspace{-0.05cm} =\hspace{-0.05cm} \sum_{j \neq j^*} \big|\supp\big(\WBjjs\big)\big|\hspace{-0.05cm} =\hspace{-0.05cm} \ell(q^\ell - q^{\ell-s}),
\]
where the last equality is due to Theorem~\ref{thm:fixed_basis}, regardless of the choice of basis $\B$ at Node~$j^*$. Thus, $\cI(\R) = \ell(q^\ell - q^{\ell-s})$. 
\end{proof} 

Since $\Sjsj = \S/(\bal_{j^*} - \bal_j) = \S/(\bal_j-\bal_{j^*}) = \Sjjs$, the collection of repair schemes proposed in~\cite[Construction III]{DauMilenkovic2017} is symmetric. Hence, even if different storage nodes are allowed to optimize their own bases, the average I/O cost is still $k\ell$. 
The same conclusion holds for~\cite[Construction II]{DauMilenkovic2017}. 

\begin{corollary} 
The collection of repair schemes for full-length Reed-Solomon codes with $n = q^\ell$ and $r = n - k = q^m$, $1 \leq m < \ell$, proposed in~\cite{DauMilenkovic2017}, have the average I/O cost $k\ell$.
\end{corollary} 

\section{Bandwidth Optimality Requires High I/O Cost}
\label{sec:converse}

In this section, we show that when $r = q = 2$, every bandwidth-optimal linear repair scheme for a full-length RS code over $\fql$ must be rotational, which 
in turn implies that high I/O cost is \emph{necessary} to achieve optimal bandwidth. Note that as proved in~\cite{GuruswamiWootters2016}, every linear repair scheme for an RS code can be described as in Section~\ref{sec:pre}.
A fixed basis is assumed. 

A characterization of rotational repair schemes with column-dimension $s = \ell-1$ is presented in Lemma~\ref{lem:rotational_characterization}. 

\begin{lemma} 
\label{lem:rotational_characterization}
A linear repair scheme for a full-length RS code over $\fql$ is rotational with column-dimension $\ell-1$ if and only if every $(\ell-1)$-dimensional subspace of $\fql$ appears among the column-spaces of the scheme exactly $q-1$ times. 
\end{lemma}
\begin{proof} 
Note that there are precisely $(q^\ell-1)/(q-1)$ $\fq$-subspaces of $\fql$ of dimension $\ell-1$.
Therefore, for any $\fq$-subspace $\S$ of dimension $\ell-1$, the collection of $\fq$-subspaces $\{\bga \S \colon \bga \in \fql^*\}$ covers each $(\ell-1)$-dimensional $\fq$-subspace of $\fql$ precisely $q-1$ times.  
That explains the lemma. 
\end{proof}

Lemma~\ref{lem:simplified} states the fact that for full-length RS codes, 
to study repair bandwidth and I/O cost, it suffices to just examine repair schemes for the first component $\bco = f(0)$, $\deg(f) < k$.
This will significantly simplify our study.  
Recall that the dual of a full-length RS code is another RS code with dimension $r$. 

\begin{lemma} 
\label{lem:simplified}
Let $\bgi = \big(g_i(\bal_1),\ldots,g_i(\bal_n)\big)$, $i\in [\ell]$, where $n = q^\ell$, $\fql = \{0=\bal_1,\bal_2,\ldots,\bal_n\}$, and $g_i(x) \in \fql[x]$ are polynomials of degree at most $r-1$. 
Let $h_i(x) \define g_i(x+\bal_{j^*})$ and $\bhi = \big(h_i(\bal_1),\ldots,h_i(\bal_n)\big)$. Then $\{\bgi\}_{i=1}^\ell$ forms a repair scheme for $\bcjs$ if and only if 
  $\{\bhi\}_{i=1}^\ell$ forms a repair scheme for $\bco$ and moreover, these two schemes will have the same repair bandwidth and I/O cost. 
\end{lemma} 
\begin{proof} 
Since $h_i(\bal_1) = h_i(0) = g_i(\bal_{j^*})$, the set $\{g_i(\bal_{j^*})\}_{i=1}^\ell$ is an $\fq$-basis of $\fql$ if and only if the set $\{h_i(\bal_1)\}_{i=1}^\ell$ is an $\fq$-basis of $\fql$. This explains the first statement of the lemma. 
For the second statement on repair bandwidth and I/O cost, note that as the code is full length, we have $\{\bal_1,\ldots,\bal_n\} \equiv \fql$. 
Moreover, $h_i(\bal_j) = g_i(\bal_j + \bal_{j^*})$, for $j \in [n]$. 
Hence, the collection of column-spaces of the repair scheme for $\bco$ based on $\{\bhi\}_{i=1}^\ell$ is simply a rearrangement of the column-spaces of the repair scheme for $\bcjs$ based on $\{\bgi\}_{i=1}^\ell$.
\end{proof} 

\begin{lemma} 
\label{lem:basis_transform}
Suppose $B = \{\bb_i\}_{i=1}^\ell$ is an $\fq$-basis of $\fql$ while
$A = \{\ba_i\}_{i=1}^\ell$ is not.
Then there exists $\bga \in \fql^*$ so that $A + \bga B \define \{\ba_i+\bga\bb_i\}_{i=1}^\ell$ is also an $\fq$-basis of $\fql$.  
\end{lemma}
\begin{proof} 
Set $\S_A = \mathsf{span}_{\fq}(A)$ and $\tau$ a mapping from $\fql$ to $\S_A$ defined by $\tau(\sum_{i=1}^\ell \eta_i\bb_i) = \sum_{i=1}^\ell \eta_i\ba_i$, for every $\eta_i \in \fq$. 
Since $A$ is linearly dependent over $\fq$, there exists $\bu' \in \fql^*$ such that $\tau(\bu') = 0$. Therefore, the set $C \define \{-\bu^{-1}\tau(\bu) \colon \bu \in \fql^*\}$ contains $0$. Moreover, it is clear that $|C| \leq |\fql^*| = q^\ell-1$. Therefore, there exists a nonzero element $\bga \not\in C$. We now show that $A + \bga B$ is linearly independent over $\fq$. 
Indeed, it suffices to show that for every $(\eta_1,\ldots,\eta_\ell) \not\equiv
(0,\ldots,0)$, we have $\sum_{i=1}^\ell \eta_i (\ba_i + \bga\bb_i) \neq 0$.
Let $\bu = \sum_{i=1}^\ell \eta_i\bb_i$, then because $B$ is a basis, $\bu \neq 0$. As $\bga \not\in C$, we have $\bga\bu \neq -\tau(\bu)$, which implies that 
$\sum_{i=1}^\ell \eta_i (\ba_i + \bga\bb_i) \neq 0$, as desired.
\end{proof} 

Lemma~\ref{lem:column_space} is due to Proposition~1 and Corollary~1 in~\cite{DauMilenkovic2017}.\vspace{-3pt}

\begin{lemma} 
\label{lem:column_space}
In every bandwidth-optimal linear repair scheme for a full-length Reed-Solomon code with $n = q^\ell$, $r = q^m$, and $m \in [\ell-1]$, 
the column-spaces all have dimension $\ell-m$. 
\end{lemma}

We are now ready to prove the main theorem of this section. 
\begin{theorem} 
\label{thm:converse}
Every bandwidth-optimal linear repair scheme for full-length Reed-Solomon codes with $n = 2^\ell$ and $r = 2$ must be rotational. Thus, such a scheme
must incur an I/O cost $k\ell$. 
\end{theorem}
\begin{proof}
By Lemma~\ref{lem:simplified}, it suffices to consider a bandwidth-optimal repair scheme for $\bco$. Supposed that this scheme is based on the dual codewords $\bgi = \big(g_i(\bal_1),\ldots,g_i(\bal_n)\big)$, $i\in [\ell]$, where $n\hspace{-0.05cm} =\hspace{-0.05cm} 2^\ell$, $\ftl\hspace{-0.05cm} =\hspace{-0.05cm} \{0=\bal_1,\bal_2,\ldots,\bal_n\}$, and $g_i(x) \in \ftl[x]$ are polynomials of degree at most one. Set $\bb_i = g_i(0)$, then $\{\bb_i\}_{i=1}^\ell$ is an $\ft$-basis of $\ftl$ since $\{\bgi\}_{i=1}^\ell$ forms a repair scheme for $\bco$. As the scheme is bandwidth-optimal, by Lemma~\ref{lem:column_space}, 
$\dim_{\ft}(\Sjo) = \ell-1$, for every $j \neq 1$.

To prove by contradiction, we assume that the scheme is \emph{not} rotational.
Due to Lemma~\ref{lem:rotational_characterization}, as $q = 2$, this means that there exist two identical column-spaces. Without loss of generality, we may assume that $\S_{23} \define \Sto \equiv \Stho$. 
Then $g_i(\bal_2) = \git \in \S_{23}$ and $g_i(\bal_3) = \gith \in \S_{23}$, for $i \in [\ell]$. Note that $\dim_{\ft}(\S_{23}) = \ell-1$. 
By interpolation,
\[
g_i(x) = \dfrac{\bal_3\git + \bal_2\gith}{\bal_2+\bal_3} 
+ \dfrac{\git + \gith}{\bal_2 + \bal_3}x = \bb_i + \ba_i \dfrac{x}{\bal_2+\bal_3},
\]
where $\ba_i \define \git + \gith \in \S_{23}$ and $\bb_i = g_i(0)$. 
Then $A\hspace{-0.05cm} =\hspace{-0.05cm} \{\ba_i\}_{i=1}^\ell\hspace{-0.05cm} \subseteq\hspace{-0.05cm} \S_{23}$ and $B\hspace{-0.05cm} =\hspace{-0.05cm} \{\bb_i\}_{i=1}^\ell$ satisfy the condition of Lemma~\ref{lem:basis_transform}. Hence, there exists $\bga\hspace{-0.05cm} \neq\hspace{-0.05cm} 0$ such that $A + \bga B$ is an $\ft$-basis of $\ftl$, which implies that $B + \bga^{-1}A$ is also a basis. Take $j \in [n]$ such that $\frac{\bal_j}{\bal_2+\bal_3} = \bga^{-1}$. Then $\bal_j \neq \bal_1 = 0$ and \vspace{-7pt}
\[
\Sjo\hspace{-0.1cm} =\hspace{-0.05cm} \mathsf{span}_{\ft}\hspace{-0.05cm}\big(\hspace{-0.05cm}\{g_i(\bal_j)\}_{i=1}^\ell\big)\hspace{-0.1cm}
=\hspace{-0.05cm} \mathsf{span}_{\ft}\hspace{-0.05cm}\big(\hspace{-0.05cm}\{\bb_i\hspace{-0.05cm} +\hspace{-0.05cm} \bga^{-1}\ba_i\}_{i=1}^\ell\big)\hspace{-0.06cm} =\hspace{-0.05cm} \ftl. \vspace{-3pt}
\]
This contradicts the earlier statement that $\dim_{\ft}(\Sjo) = \ell-1$ whenever $j \neq 1$. Thus, such a scheme must be rotational. The conclusion on the I/O cost follows from Theorem~\ref{thm:fixed_basis}.
\end{proof} 
\vspace{-3pt}

Finally, we remark that the conclusion of Theorem~\ref{thm:converse} does not extend to full-length MDS codes. Indeed, one can easily find a repair scheme  for a $[4,2]_4$ MDS code that is bandwidth optimal but not rotational. For instance, take $\bg^{(1)} = (1,0,1,1)$ and $\bg^{(2)} = (\bxi, 1, 0, 1)$, where $\ff_4 = \{0,1,\bxi,\bxi+1\}$. 
\vspace{-2pt}

\section*{Acknowledgment}
The authors thank Dung Duong for helpful discussions. 
This work is supported by the 210124 ARC DECRA grant (DE180100768), the Monash 250003 Faculty Initiative Fund, and the NSF grant CCF 1619189. 

\bibliographystyle{IEEEtran} \vspace{-3pt}
\bibliography{ReadCost}

\end{document}